\documentclass[12pt]{article}
\usepackage{lmodern}        

\usepackage[T1]{fontenc}
\usepackage[latin9]{inputenc}
\usepackage[top=1in, bottom=1in, left=1in, right=1in]{geometry}
\usepackage{verbatim}
\usepackage{float}
\usepackage{amssymb,amsfonts,amsmath}
\usepackage{enumerate}
\usepackage{enumitem}
\usepackage{amsthm}
\usepackage{threeparttable,hhline}
\usepackage{algorithm}
\usepackage{algpseudocode}
\usepackage{multirow, array}

\makeatletter
\renewcommand\subsubsection{\@startsection{subsubsection}{2}%
  \z@{.5\linespacing\@plus.7\linespacing}{-.5em}%
  {\normalfont\bfseries}}
\makeatother

\usepackage{color}
\usepackage{graphicx}

\usepackage{amsthm}

\usepackage{hyperref}
\usepackage{caption}
\usepackage{subcaption}
\floatstyle{ruled}
\newfloat{Figure}{tbp}{loa}
\floatname{Figure}{Figure}

\usepackage{microtype}

\DeclareGraphicsExtensions{.eps}

\newcommand\nc\newcommand
\nc\bfa{{\boldsymbol a}}\nc\bfA{{\boldsymbol A}}\nc\cA{{\mathcal A}}
\nc\bfb{{\boldsymbol b}}\nc\bfB{{\boldsymbol B}}\nc\cB{{\mathcal B}}
\nc\bfc{{\boldsymbol c}}\nc\bfC{{\boldsymbol C}}\nc\cC{{\mathcal C}}
\nc\sC{{\mathscr C}}
\nc\bfd{{\boldsymbol d}}\nc\bfD{{\boldsymbol D}}\nc\cD{{\mathcal D}}
\nc\bfe{{\boldsymbol e}}\nc\bfE{{\boldsymbol E}}\nc\cE{{\mathcal E}}
\nc\bff{{\boldsymbol f}}\nc\bfF{{\boldsymbol F}}\nc\cF{{\mathcal F}}
\nc\bfg{{\boldsymbol g}}\nc\bfG{{\boldsymbol G}}\nc\cG{{\mathcal G}}
\nc\bfh{{\boldsymbol h}}\nc\bfH{{\boldsymbol H}}\nc\cH{{\mathcal H}}
\nc\bfi{{\boldsymbol i}}\nc\bfI{{\boldsymbol I}}\nc\cI{{\mathcal I}}
\nc\bfj{{\boldsymbol j}}\nc\bfJ{{\boldsymbol J}}\nc\cJ{{\mathcal J}}
\nc\bfk{{\boldsymbol k}}\nc\bfK{{\boldsymbol K}}\nc\cK{{\mathcal K}}
\nc\bfl{{\boldsymbol l}}\nc\bfL{{\boldsymbol L}}\nc\cL{{\mathcal L}}
\nc\bfm{{\boldsymbol m}}\nc\bfM{{\boldsymbol M}}\nc\sM{{\mathscr M}}\nc\cM{{\mathcal M}}
\nc\bfn{{\boldsymbol n}}\nc\bfN{{\boldsymbol N}}\nc\cN{{\mathcal N}}
\nc\bfo{{\boldsymbol o}}\nc\bfO{{\boldsymbol O}}\nc\cO{{\mathcal O}}
\nc\bfp{{\boldsymbol p}}\nc\bfP{{\boldsymbol P}}\nc\cP{{\mathcal P}}
\nc\bfq{{\boldsymbol q}}\nc\bfQ{{\boldsymbol Q}}\nc\cQ{{\mathcal Q}}
\nc\bfr{{\boldsymbol r}}\nc\bfR{{\boldsymbol R}}\nc\cR{{\mathcal R}}
\nc\bfs{{\boldsymbol s}}\nc\bfS{{\boldsymbol S}}\nc\cS{{\mathcal S}}
\nc\bft{{\boldsymbol t}}\nc\bfT{{\boldsymbol T}}\nc\cT{{\mathcal T}}
\nc\bfu{{\boldsymbol u}}\nc\bfU{{\boldsymbol U}}\nc\cU{{\mathcal U}}
\nc\bfv{{\boldsymbol v}}\nc\bfV{{\boldsymbol V}}\nc\cV{{\mathcal V}}
\nc\bfw{{\boldsymbol w}}\nc\bfW{{\boldsymbol W}}\nc\cW{{\mathcal W}}
\nc\bfx{{\boldsymbol x}}\nc\bfX{{\boldsymbol X}}\nc\cX{{\mathcal X}}
\nc\bfy{{\boldsymbol y}}\nc\bfY{{\boldsymbol Y}}\nc\cY{{\mathcal Y}}
\nc\bfz{{\boldsymbol z}}\nc\bfZ{{\boldsymbol Z}}\nc\cZ{{\mathcal Z}}

\nc\diff{{\mathrm d}}
\nc\e{{\mathrm e}}
\nc\calC{{\mathcal C}}

\DeclareMathOperator{\rank}{rank}

\DeclareMathOperator{\wt}{wt}

\newcommand{\remove}[1]{}

\newcommand{\avg}{{\mathbb E}}

\newcommand{\dist}{d_\mathrm{H}}

\newtheorem{theorem}{Theorem}
\newtheorem{proposition}{Proposition}
\theoremstyle{definition}
\newtheorem{definition}{Definition}
\theoremstyle{corollaryn}

\theoremstyle{theorem-n}
\newtheorem*{theorem-n}{Theorem}
\newtheorem*{lemma-n}{lemma}

\newcommand{\capa}{{\sf Cap}}

\def\DEBUG{true}

\ifdefined\DEBUG

  \def\rem#1{{\marginpar{\raggedright\scriptsize #1}}}
  \newcommand{\barnr}[1]{\rem{\textcolor{red}{$\bullet$ #1}}}
  \newcommand{\aryar}[1]{\rem{\textcolor{green}{$\bullet$ #1}}}
\else

  \newcommand{\barnr}[1]{}
  \newcommand{\aryar}[1]{}

\fi

\newcommand\ff{{\mathbb F}}

\allowdisplaybreaks

\author{Arya Mazumdar}

\begin{document}
\sloppy

\title{Capacity of Locally Recoverable Codes\thanks{Halicioglu Data Science Institute, University of California, San Diego \url{arya@ucsd.edu}. This work is supported in part by an NSF awards CCF 2127929  and  CCF  1618512. Some parts of this paper was presented at IEEE Information Theory Workshop, 2018, as an invited paper.}}
\maketitle

\begin{abstract}
Motivated by applications in distributed storage, the notion of a locally recoverable code (LRC) was introduced a few years back. In an LRC, any coordinate of a codeword is 
recoverable by accessing only a small number of other coordinates. While different properties of LRCs have been well-studied, their performance on channels with random erasures or errors has been mostly unexplored. 
In this paper, we analyze the performance of LRCs over such stochastic channels. In particular, for input-symmetric discrete memoryless channels, we give a tight characterization
of the gap to Shannon capacity when LRCs are used over the channel. Our results hold for a general notion of LRCs that correct multiple local erasures.
\end{abstract}

\section{Introduction}

A code $\cC$, a collection of vectors,  is called locally recoverable (or repairable) with locality $r$, if content of any coordinate can be recovered by accessing only $r$ other coordinates \cite{gopalan2012locality,papailiopoulos2012locally}.
Locally recoverable codes have been the subject of  intense research, including constructions \cite{tamo2014family,barg2017locally,tamo2016optimal,prakash2012optimal}, bounds \cite{cadambe2015bounds,agarwal2018combinatorial,tamo2016bounds} and generalizations \cite{wang2014repair,rawat2015cooperative,rawat2014locality,mazumdar2015storage,mazumdar2019storage,tamo2014bounds,prakash2012optimal,karingula2022lower}.

Formally, a $q$-ary code $\cC$ of length $n,$ cardinality $M,$ and  distance $d$ is a set of $M$ length-$n$ vectors over an alphabet $Q, |Q|=q$, with minimum pairwise Hamming distance $d$. The quantity $k=\log_qM$ is called the dimension of $\cC,$ and $R = \frac1n \log_qM$ is called the rate of the code. If $Q$ is a finite field and
$\cC$ is a linear subspace of $Q^n$ then $k$ is the dimension of $\cC$ as a vector space. 
Below, $[n] \equiv \{1,\dots, n\}$, and for any $x \in Q^n$, $x_{i}$ is the projection of $x$ in the $i$th coordinate. By extension, for any $I \subseteq [n]$, $x_I$ is the projection of $x$ onto the coordinates of $I$.

\begin{definition}\label{def1}
  A code $\cC\subset Q^n$ is \emph{locally recoverable code} (LRC) with {\em locality} $r$
if every coordinate $i\in \{1,2,\dots,n\}$ is contained in a subset $\cR_i\subseteq[n]$ of size $r+1$ such that
 there is a function
 $\phi_i:Q^r\to Q$ with the property that for every codeword $c=(c_1, c_2, \dots, c_n)\in\cC$
   \begin{equation}\label{eq:def1}
   c_i=\phi_i(c_{j_1},\dots,c_{j_r}),
   \end{equation}
where $j_1 < j_2 < \cdots < j_r$ are the elements of $\cR_i\backslash\{i\}.$
We use the notation $(n,k,r)$ to refer to a code of length $n$, dimension $k$ and locality $r.$
\end{definition}

This definition has been extended
in \cite{prakash2012optimal}, which has been since widely used e.g.~\cite{cai2019optimal,wang2014repair,chen2017constructions}.
\begin{definition}\label{def2} A code $\cC\subset Q^n$ of cardinality $q^k$ is said to have the $(\rho,r)$ {\em locality property} (to be an $(n,k,r,\rho)$ LRC) where $\rho\geq 2$, 
if each coordinate $i\in [n]$ is contained in a subset $\cR_i\subset [n]$ of size at most $r+\rho -1$ 
such that the restriction $\cC_{\cR_i}$ of the code $\cC$ to the coordinates in $\cR_i$ forms a code of distance at least $\rho$. 
Notice that the values of any $\rho-1$ coordinates of $\cR_i$ are determined by  the values  of the  remaining $|\cR_i|-(\rho-1)\leq r$ coordinates, thus enabling local recovery.  $\cR_i$ is called the repair group of coordinate $i$.
\label{def:LRC}
\end{definition}

Over the past decade, many features of locally recoverable codes were examined, most notably  the minimum distance of LRC codes, e.g.~\cite{gopalan2012locality,cadambe2015bounds,tamo2016bounds}. However, optimal lengths and symbol-size were also well-studied~\cite{kolosov2018optimal,guruswami2019long,cai2020optimal}. Initially, research focused on codes with large alphabets, but interest has also shifted to studying binary and other small alphabet LRCs, both in terms of the bounds and constructions possible with the use of cyclic and algebraic properties of codes~\cite{huang2016binary,goparaju2014binary,jin2022binary,barg2017locallyr}. It has been suggested that binary codes are efficient for storage, and while the origins of LRCs can be traced back to distributed storage systems, they have also become a topic of theoretical interest on their own. Many recent studies have focused on developing bounds and constructions of codes with local repair properties, not only for their potential use in distributed storage, but also as an intellectual exercise.

Aside from constructing LRCs, much of the research in this field has focused on determining the optimal error-correction capability of LRCs. Typically, the error/erasure correction capability of a code is represented by its minimum distance, under the assumption of an adversarial error model. However, an arguably more common scenario is one in which errors and erasures occur randomly. Despite this, there has been relatively little research conducted in that direction.

In this paper, we investigate the maximum achievable rate of  locally repairable codes such that reliable transmission is possible over a discrete memoryless channel (DMC). Surprisingly, with the exception for \cite{mazumdar2014update}, no paper deals with this quite basic theoretical question. In \cite{mazumdar2014update} it was shown that for  a {\em binary erasure channel} (BEC) with erasure probability $p$ (Shannon capacity $1-p$), to achieve a rate of $1-p-\epsilon$, the locality must scale as $\Theta(\log \frac1\epsilon)$. While the constant within $\Theta(\cdot)$ is not clear, the method therein also does not extend to 
{\em binary symmetric channel} (BSC) or other binary-input memoryless channels or the generalized notion of LRCs. 

In this work, we perform a finer and through analysis of the gap to capacity for LRCs. For a discrete memoryless channel given by a input-output stochastic transition matrix\footnote{We sometime also refer to a DMC by $X \to Y$ to describe the input-output random variables.} $W$, let $\capa(W)$ be the Shannon capacity of the channel, and $\capa(W,r)$ to be the capacity of the channel where we are constrained to use only a locally repairable code with locality $r$, and $\rho=2$ (Definition~\ref{def1}). Let us define,
$$
{\rm Gap}(W,r) \equiv \capa(W) - \capa(W,r).
$$
An impossibility result in this regard gives a lower bound on the gap, while an achievability scheme gives an upper bound on the gap. Our results for LRCs with parameter $r$ (Definition~\ref{def1}) are summarized in Table~\ref{tab:results}. Here, $h(x) \equiv -x\log_2 x -(1-x)\log_2(1-x)$ is the binary entropy function. While the results hold for binary-input channels, it is not difficult to extend the for the $q$-ary case. For the BEC and BSC, the results are also plotted in 
Fig.~\ref{fig:cap} for $r=2$. Note that, we are able to exactly calculate the capacity for BEC, while we have tight upper and lower bounds for BSC.

\begin{center}
\begin{threeparttable}[!t]
\caption{The gap to capacities of LRCs over binary-input symmetric DMCs\label{tab:results} for $\rho=2$}
\label{results_table}
\begin{centering}
\begin{tabular}{|m{30mm}|c|c|}
\hline
 \centering \textbf{Channel} & \textbf{Lower Bound} on ${\rm Gap}(W,r)$ &  \textbf{Upper Bound} on ${\rm Gap}(W,r)$\\
 \hhline{|=|=|=|}
\centering BEC($p$) & $\frac{(1-p)^{r+1}}{r+1}$ & $\frac{(1-p)^{r+1}}{r+1}^\ast$  \\
 \hline
 \centering BSC($p$) & $\frac{(1-h(p))^{r+1}}{r+1}$ & $\frac1{r+1}\Big(1-h\Big(\frac{1-(1-2p)^{r+1}}{2}\Big)\Big)^{\dagger}$  \\
 \hline
\centering General $W$ & $ \frac{\capa(W)^{r+1}}{r+1}$ & $\frac1{r+1}\Big(1-h\Big(\frac{1-(1-2h^{-1}(1-\capa(W)))^{r+1}}{2}\Big)\Big)$  \\
\hline
\end{tabular}
\begin{tablenotes}
\footnotesize
\item $^\ast$ also achievable by linear codes.
\item $^\dagger$ we conjecture this bound to be tight. 
\end{tablenotes}
\end{centering}
\end{threeparttable}
\end{center}

It is of interest to compare the results of Table~\ref{tab:results} with the results on rate vs. minimum distance trade-off for LRCs from existing literature. Note that, even in standard error-correcting codes, and in asymptotics, the upper and lower bounds on the rate vs. minimum distance trade-off do not match. This is reflected in corresponding bounds for LRCs too~\cite{cadambe2015bounds}. Moreover, for LRCs these bounds can be quite cumbersome~\cite{agarwal2018combinatorial}. One insightful (though sub-optimal) bound maybe the {\em Gilbert-Varshamov}-type bound presented in \cite[Eq.~(12)]{cadambe2015bounds}, which says that there exists a family of $(n,Rn,r,2)$-LRC  with minimum distance $\delta n,$ such that
$$
R \ge 1-h(\delta) - \frac{\log_2(1+(1-2\delta)^{r+1})}{r+1}.
$$
The ``gap'' from the actual Gilbert-Varshamov bound here is $\frac{1}{r+1}\log_2(1+(1-2\delta)^{r+1}),$ which bears similarity with the analogous terms of Table~\ref{tab:results}.

To prove the lower  and upper bounds for BEC we rely on simple information inequalities and random coding methods. The main idea behind the converse bound is that if a coordinate of a codeword and its repair group are both unerased then there is redundant information. It is difficult to extend the converse bounding arguments to other channels. However in some sense BEC is the `best' channel among all binary input memoryless symmetric channels \cite{korner1977comparison}. We can use that fact to lower bound the gap to capacity for more general channels including BSC. A random coding method for BSC also gives the upper bound on gap to capacity for any binary input channels by the same argument, as BSC is the `worst' among all in the same sense. This notion of `best' and `worst' channels are formalized and discussed in detail later.


\begin{figure}
	\centering
	\begin{subfigure}{0.45\textwidth} 
		\includegraphics[width=\textwidth]{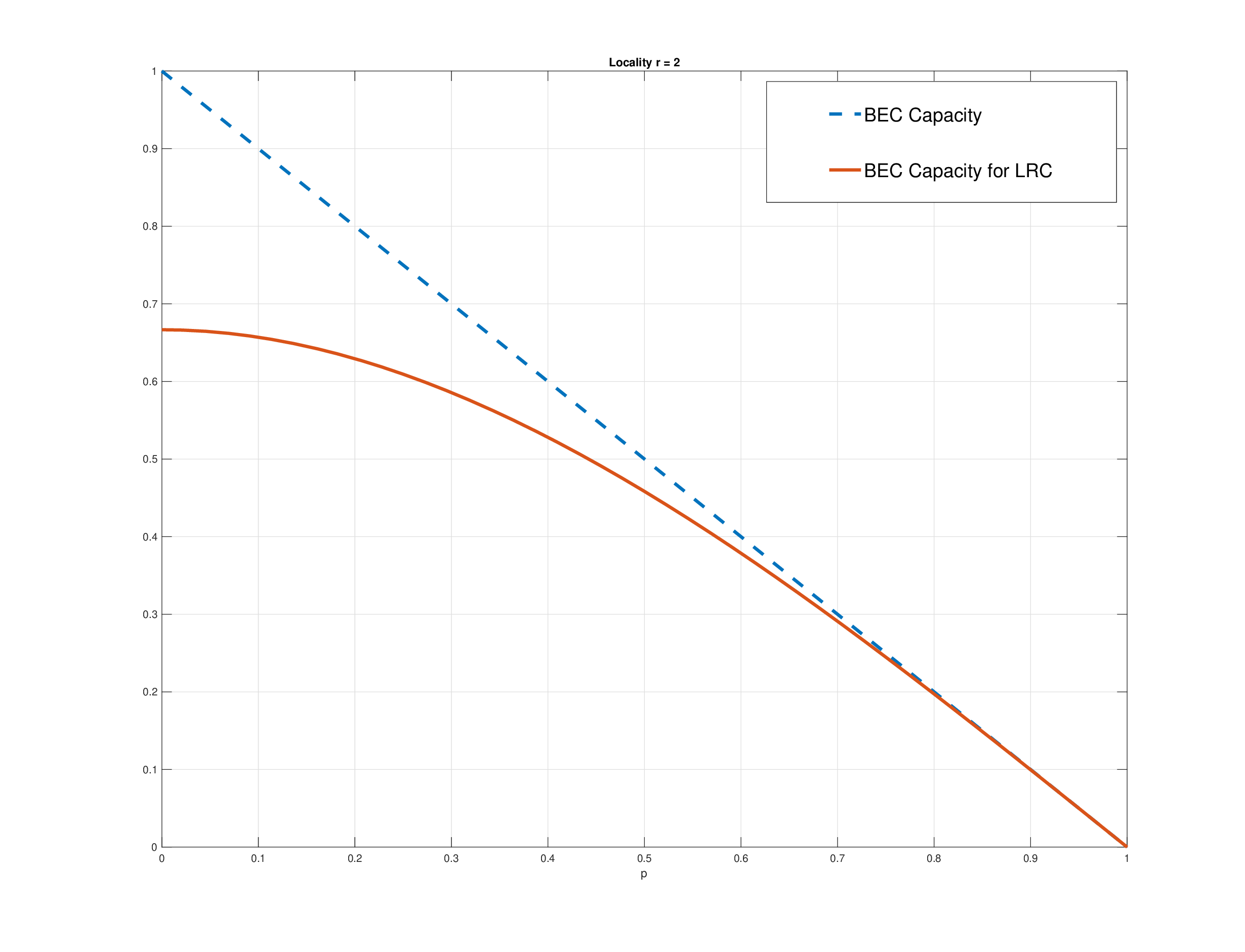}
	\end{subfigure}
	\vspace{1em} 
	\begin{subfigure}{0.45\textwidth} 
		\includegraphics[width=\textwidth]{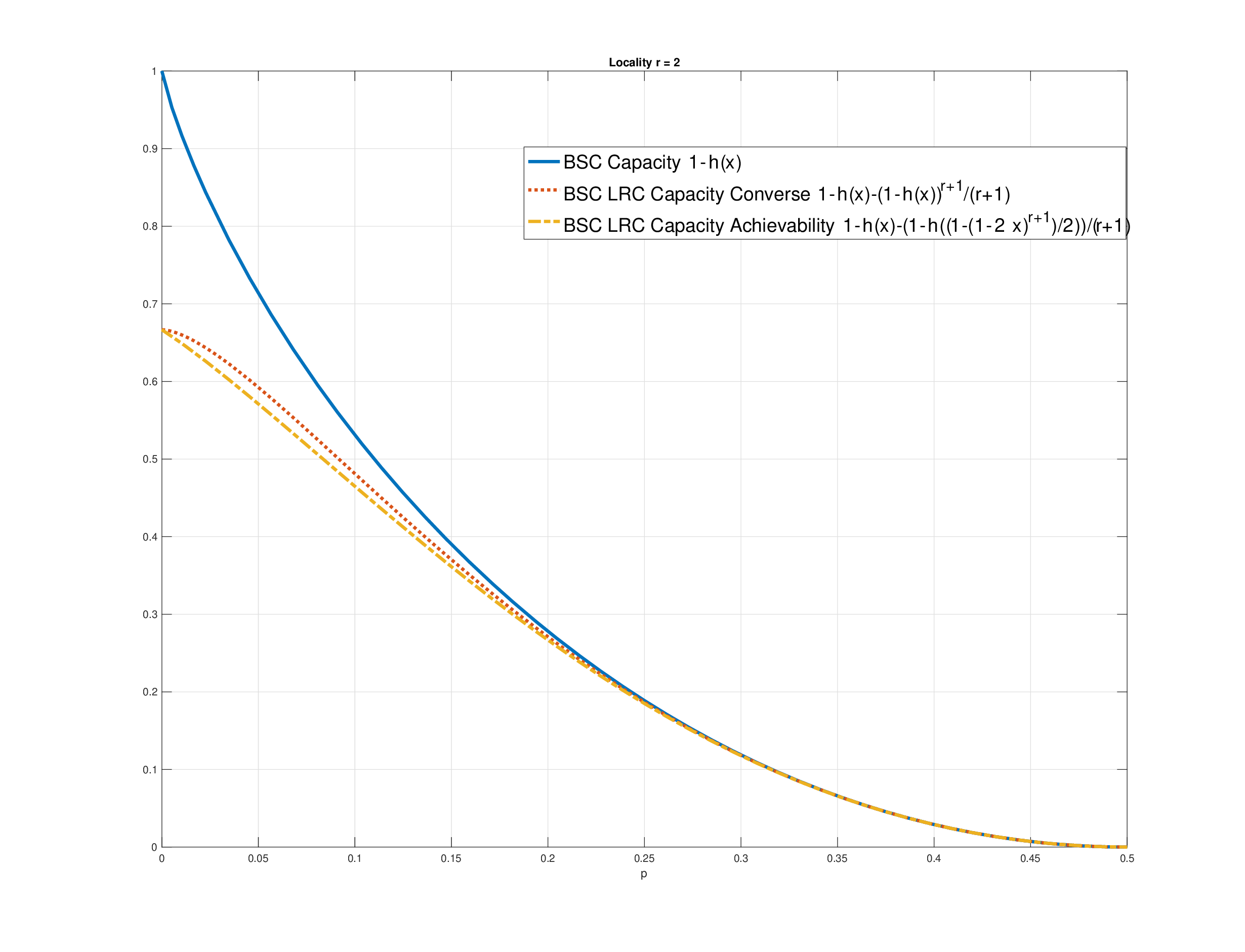}
	\end{subfigure}
	\caption{Capacities of Locally Recoverable Codes with locality $r=2$ (and $\rho=2$) over BEC and BSC. \label{fig:cap}} 
\end{figure}

We also analyze the capacity of LRCs with the more general definition (cf. Definition~\ref{def2}).
While the information theoretic methods for converse bound extends for this general case, the random coding bounds for achievability becomes trickier, and for $\rho>2$, the upper and lower bounds, even for erasure channel, starts to diverge.  Nonetheless, for $\rho=3$ we have a tight closed form expressions for both the bounds. 
The main idea for these achievability results is to use code concatenation~\cite{forney1966concatenated}. We use a small (constant) length code with a prescribed minimum distance as an inner code of the concatenated code construction. This guarantees the local repair property. The outer code is a random code. The achievable rate of this construction that guarantees a vanishing probability of error is proportional to the mutual information $I(X;Y)$ where $X$ is chosen uniformly at random from the inner code, and $Y$ is the output of the channel when $X$ is the input. For BEC, this quantity  can be concisely expressed by the  so-called rank (Tutte) polynomial of the local codes formed by the repair groups, and for BSC it depends on the coset weight distributions of the local codes. 
We provided pointers to these results in Table~\ref{tab:thms}.

\begin{table}[htbp]
\begin{center}
\caption{ \label{tab:thms}Pointers to results for capacity of $(n,k,r,\rho)$ LRC, general $\rho$}
\begin{tabular}{ |c|c|c| }
\hline
{\bf Channel} & {\bf Achievability} & {\bf Converse}\\
 \hline\hline
 BEC($p$) & Theorem~\ref{thm:main3} & Theorem~\ref{thm:main2} \\ 
 \hline
 BSC($p$) & Theorem~\ref{thm:bsc3} & Theorem~\ref{thm:bsc2} \\ 
 \hline
\end{tabular}
\end{center}
\end{table}
The paper is organized as follows. In Section~\ref{sec:con}, we describe some concepts and definitions that are going to be used in the rest of the paper. 
Sections \ref{sec:bec} and  \ref{sec:bsc} deal with the binary erasure and binary symmetric channels respectively, while Sec.~\ref{sec:all} deals with other binary input channels.  

\section{Some Coding Theoretic Concepts}\label{sec:con}
In this paper,  we consider only binary codes. A code of length $n$ is usually denoted by $\cC \in \{0,1\}^n$.  
Let $A(n,d)$ be the maximum possible size of a code $\cC \subseteq \{0,1\}^n$ of minimum distance $d$. Let the least possible redundancy in a code of length $n$ and distance $d$ be denoted by $\mu(n,d) \equiv n - \log_2 A(n,d)$.

We will use the standard information theoretic notion of channel capacity. Let $\cC$ be a code and $X_1^n \equiv (X_1,\dots, X_n)$ be a randomly and uniformly chosen codeword (we write $X_1^n \sim {\rm Unif}(\cC)$). Let $Y_1^n \in \cY^n$ be the output of the discrete memoryless channel $W$ when $X_1^n$ is the input. Suppose $f: \cY^n \to \cC$ be a decoding algorithm for $\cC$. The average probability of error is defined to be:
$$
P(\cC) = \avg_{X_1^n \sim {\rm Unif}(\cC)} \Pr(f(Y_1^n) \ne X_1^n).
$$
The capacity of the binary input discrete-memoryless channel $W$ is defined to be,
$$
\capa(W) = \inf_{\epsilon >0}\limsup_{n \to \infty} \max_{\cC \subseteq \{0,1\}^n : \exists f {\rm s.t.} P(\cC) < \epsilon} \frac{\log_2 |C|}{n}
$$
The capacity of LRCs is defined in the similar way:
$$
\capa(W,r, \rho) = \inf_{\epsilon >0}\limsup_{n \to \infty} \max_{\substack{\cC \subseteq \{0,1\}^n : \exists f {\rm s.t.} P(\cC) < \epsilon\\ \cC \text{ is } (n,k,r, \rho) \text{ LRC}}} \frac{k}{n}
$$

For BEC and BSC with parameter $p$, the capacities are $1-p$ and $1-h(p)$ respectively. We use some shorthands. For BEC and BSC with parameter $p$, we write the respective LRC capacities as $\capa_{\rm BEC}(p,r,\rho)$ and  $\capa_{\rm BSC}(p,r,\rho)$.

We will also be needing the notion of the {\em rank polynomial} of a linear code for our results. Suppose $\cC$ is a linear  code of length $n$ and dimension $k$. This means $\cC$ has a $k \times n$ generator matrix $G$ of rank $k$. For a subset of indices $I \subseteq \{1,2,\dots, n\}$,  let $G_I$ denote the $k\times |I|$ submatrix of $G$ that contains the columns with indices in $I$. Let $E(u,v)$ be the number of $k \times u$ submatrices of $G$ with rank $v$, i.e.,
$$
E(u,v) = |\{I: |I| = u, \text{ rank of } G_I =v\}|.
$$
The {\em rank polynomial} of the code $\cC$ is defined by~\cite{barg2002some},
$$
U(x,y) = \sum_{u=0}^n \sum_{v=0}^k E(u,v) x^u y^v.
$$
As described in \cite{barg2002some}, this polynomial is closely related to the Tutte polynomial of the vector matroid of the code, and as such satisfies some nice properties.

Since a linear code $\cC\subseteq \{0,1\}^n=\ff_2^n$ is a subgroup of the additive group of the vector space, the translates or cosets of the code partitions $\ff_2^n$ and are of equal size. There are $2^n/|\cC|$ cosets, denoted by $\cC^{i}, i=0, 2, \dots, 2^n/|\cC|-1$. Let  $A^{(i)}_w$ be the number of vectors of Hamming weight $w$ in the coset $\cC^{i}$, $i=0,1,  2, \ldots, \frac{2^{n}}{|\cC|}-1$.  The {\em coset weight enumerator} of the code is defined by:
$$
A^{(i)}(x,y) = \sum_{w=0}^{n} A^{(i)}_w  x^{n-w} y^w.
$$

Last, but not the least, we need the concept of ``more capable channels''. All the channels below are discrete memoryless channels.

\begin{definition}
A  channel $X \to Y$ is said to be more capable than another channel $X\to Z$ if for any input distribution on $X$, $$I(X;Y) \ge I(X;Z).$$ 
\end{definition}

It is known that among the binary-input symmetric discrete memoryless channels of same capacity
BSC is the least capable  and BEC is the most capable~\cite{geng2013broadcast}. The following can be derived from~\cite{korner1977comparison}. This result also follows from \cite[ex. 16, p.~116]{csiszar1982}.
\begin{proposition}\label{prop:cap}
Suppose the channel  $X \to Y$ is more capable than the channel $X\to Z$, and a code $\cC$ of rate  $R$ achieves a probability of error $\epsilon$ over the channel  $X \to Z$. Then there exists a code $\cC' \subseteq \cC$ of rate $R -\delta$ that achieves a probability of error $\epsilon'$ over $X\to Y$, where $\delta, \epsilon' \to 0$ as $\epsilon \to 0$.
\end{proposition}

\section{LRC Capacity of the Binary Erasure Channel}\label{sec:bec}

Our first result concerns the LRC capacity for the special case of $\rho=2$, i.e., the usual LRCs. In this case we can exactly compute the capacity of LRCs on the BEC.
\begin{theorem}\label{thm:main1}
The capacity of LRC with locality  $r$ over BEC($p$) is given by:
$$
\capa_{\rm BEC}(p,r,2) = 1- p -\frac{(1-p)^{r+1}}{r+1}.
$$
\end{theorem}
This theorem can be proved as a corollary to the following converse and achievability results.
\begin{theorem}\label{thm:main2}
For generalized LRCs, the following holds: 
$$
\capa_{\rm BEC}(p,r, \rho) \le 1- p -\frac{(1-p)^{r}}{r+\rho-1} \sum_{t=0}^{\mu(r+\rho-1, \rho)-1}   (\mu(r+\rho-1, \rho)-t)  \binom{r+\rho-1}{t} p^t (1-p)^{\rho-1 -t} .
$$
\end{theorem}
Plugging in $\rho=2$ in the above theorem, and noting that $\mu(r+1, 2) =r+1-r=1$, we obtain
 $$
\capa_{\rm BEC}(p,r, 2) \le 1- p -\frac{(1-p)^{r+1}}{r+1},
$$
which proves the converse bound for Theorem \ref{thm:main1}.

The achievability result follows next.
\begin{theorem}\label{thm:main3}
Suppose there exists a linear code $\cA$ of length $r+\rho-1$ and minimum distance $\rho$. Let $U(x,y)$ be the rank polynomial of $\cA$. Then 
$$
\capa_{\rm BEC}(p,r, \rho) \ge \frac{p^{r+\rho-1}}{r+\rho-1}\left. \frac{\partial}{\partial y} U\Big(\frac{1-p}{p},y\Big) \right|_{y=1}.
$$
\end{theorem}
Theorems \ref{thm:main2} and \ref{thm:main3} are proved later in this section.
Note that, the single parity check code of length $r+1$ has distance $2$. Plugging $\rho =2$ in the above result, Theorem \ref{thm:main3}, we should obtain a lower bound on $\capa_{\rm BEC}(p,r, 2)$. But for this we have to obtain the rank polynomial of the parity check code.
First of all, note that, for the single parity check code of length $r+1$,
\begin{align*}
E(u,v) = \begin{cases} 1 & u= r+1, v=r \\
\binom{r+1}{u} & u \le r, v=u \\
0 & \text{ otherwise }.
\end{cases}
\end{align*}
Therefore, 
\begin{align*}
U(x,y) & = \sum_{u=0}^{r+1} \sum_{v=0}^r E(u,v) x^u y^v \\
& = \sum_{u=0}^{r}  \binom{r+1}{u} x^u y^u +  x^{r+1}y^r \\
& = (1+xy)^{r+1} - (xy)^{r+1} +x^{r+1}y^r .
\end{align*}
Differentiating, we find
\begin{align*}
\left.\frac{\partial U(x,y)}{\partial y}\right|_{y=1} &= (r+1)x(1+x)^r - (r+1)x^{r+1} + rx^{r+1}\\
&=(r+1)x(1+x)^r - x^{r+1}.
\end{align*}
Using the theorem above,
$$
\capa_{\rm BEC}(p,r, 2)  \ge p^{r+1} \frac{1-p}{p}\frac{1}{p^r} -\frac{(1-p)^{r+1}}{r+1}= 1- p -\frac{(1-p)^{r+1}}{r+1},
$$
which proves the achievability part of Theorem~\ref{thm:main1}.
It turns out that the achievability result still holds when the code in question is restricted to be linear.
\begin{theorem}
For any $\varepsilon>0$, there exists a family of linear $(n, Rn, r)$ LRC codes with rate 
$$
R \ge 1- p -\frac{(1-p)^{r+1}}{r+1} - \varepsilon,
$$
that when used over a BEC($p$) results in a probability of error that goes to $0$ with $n$.
\end{theorem}
\begin{proof}
To see this, randomly choose a $k\times n$ generator matrix in the following way. Partition the set of $n$ coordinates into $\frac{n}{r+1}$ groups of size $r+1$ each. For each group chose $r$ columns randomly and uniformly from $\{0,1\}^k$. The $r+1$st column of each group is just the coordinate-wise modulo-2 sum of all the other $r$ columns of the group. If $n$ is not divisible by $r+1$, then neglect the remainder $\leq r$ coordinate, i.e., repeat the same symbol (0 or 1) in those coordinates. Since this will not lead to an asymptotic reduction in rate, let us assume that $r+1$ divides $n$.

This random generator matrix defined a random ensemble of locally repairable codes. Let $P(\cC)$ defines the probability of error of using code $\cC$ over BEC($p$). Note that, there will be an error in decoding only when the coordinates not erased by the channel has rank (over $\ff_2$) strictly less than $k$. If we can show that $\avg_\cC P(\cC) \to 0$, i.e., the average (over all linear codes in the ensemble) probability of error goes to 0,  then there must exist codes for which the probability of error goes to 0.

Let $I\subseteq \{1, \dots, n\}$ denote the set of non-erased coordinates, and $G_I$ denote the submatrix of $G$ with only columns indexed by $I$.  Further let $Z_I$ be the number of groups from where all the $r+1$ elements are not erased. We have, for any $\epsilon >0$,
\begin{align*}
 P(\cC) &=   \sum_{u \subseteq \{1, \dots, n\}} P_{\rm BEC}(I =u, Z_I=z)  \mathrm{1}(\rank(G_I) < k \mid I =u, Z_I=z)\\
& \le    \sum_{\substack{u \subseteq \{1, \dots, n\}\\ |u| \ge n(1-p-\epsilon)\\ z \le n(\frac{(1-p)^{r+1}}{r+1}+\epsilon)} } P_{\rm BEC}(I =u, Z_I=z)  \mathrm{1}(\rank(G_I) < k \mid I =u, Z_I=z) \\
&\qquad + P_{\rm BEC}( |I| < n(1-p-\epsilon)) + P_{\rm BEC}(Z_I > n(\frac{(1-p)^{r+1}}{r+1}+\epsilon)).
\end{align*}
Note that, the last two terms of the above expression goes to $0$ exponentially with $n$ by simple application of Chernoff bound. Therefore, 
\begin{align*}
\avg_{\cC} P(\cC) &\leq        \sum_{\substack{u \subseteq \{1, \dots, n\}\\ |u| \ge n(1-p-\epsilon)\\ z \le n(\frac{(1-p)^{r+1}}{r+1}+\epsilon)} } P_{\rm BEC}(I =u, Z_I=z)  P_\cC(\rank(G_I) < k \mid I =u, Z_I=z) +o(1).
\end{align*}
The term $P_\cC(\rank(G_I) < k \mid I =u, Z_I=z)$ is simply the probability that a $k \times (|u|-z)$ random binary matrix has rank less than $k$. This probability is at most $2^{-(|u|-z -k)}$, see ~\cite[Ex.~3.21]{richardson2008modern}. Therefore,
\begin{align*}
\avg_{\cC} P(\cC) &\leq   \sum_{\substack{w \ge n(1-p-\epsilon)\\ z \le n(\frac{(1-p)^{r+1}}{r+1}+\epsilon)}}  P_{\rm BEC}(|I| =w, Z_I=z) 2^{-(w-z -k)}   +o(1),
\end{align*}
which will go to 0 exponentially with $n,$ as long as, for any $\epsilon >0$, 
$$
n(1-p - \epsilon) - n(\frac{(1-p)^{r+1}}{r+1}+\epsilon) \ge k+ \epsilon n.
$$
Rearranging the above gives the statement of this theorem. 
%
\end{proof}

In the remainder of this section we prove  Theorems~\ref{thm:main2} and \ref{thm:main3}.
\subsection{Converse Bound: Proof of Theorem~\ref{thm:main2}}
When a codeword is passed through the BEC, the non-erased coordinates must identify the sent codeword. Hence intuitively it is possible to send $n(1-p)$ bits worth information with a code of length $n$. However, when the code is also an LRC, the non-erased coordinates must contain redundant information. Indeed, if a coordinate and its repair groups are both intact, then the redundant informations in the repair group must be subtracted from to get the optimal code rate. The proof below formalizes this intuition. 
\begin{proof}[Proof of Theorem~\ref{thm:main2}]
Assume that a code $\cC, |\cC| = 2^{nR}$ is used over BEC. The random codeword $(X_1, X_2, \dots, X_n) \equiv X_1^n$ was sent over the channel. The received vector is $Z_1^n$. Let $I \subseteq \{1,2,\dots,n\}$ denote the erased coordinates.

Using Fano's inequality, the probability of error is given by,
$$
P(\cC) \ge \frac{H(X_1^n \mid Z_1^n)-1}{\log |\cC|}.
$$

Now, note that $H( I \mid X_1^n) = H(I)$. Therefore,
\begin{align*}
H(X_1^n \mid Z_1^n) & = H(X_1^n \mid Z_1^n, I) \\
& = H(X_1^n , Z_1^n, I ) - H(Z_1^n, I)\\
& = H(X_1^n) + H(Z_1^n, I \mid X_1^n)  - H(I) - H(Z_1^n \mid I)\\
& = H(X_1^n) + H( I \mid X_1^n) + H(Z_1^n \mid I,  X_1^n) - H(I) - H(Z_1^n \mid I)\\
& = H(X_1^n)  + H(Z_1^n \mid I,  X_1^n)  - H(Z_1^n \mid I) \\
& =  H(X_1^n) +0  - H(Z_1^n \mid I)\\
& =  \log|\cC|  - H(Z_1^n \mid I).
\end{align*}
This implies, 
$$
P(\cC) \ge  1 - \frac{H(Z_1^n \mid I)+1}{\log|\cC|}.
$$
Let $J(t)$ be the number of coordinates whose corresponding repair groups have $t$ coordinates erased within them. Every repair group has size at most $r+\rho-1$, which must form a code of length at most $r+\rho-1$ and distance $\rho$. Therefore the number of redundant bits within a repair group is at least $\mu(r+\rho-1, \rho) = r+\rho-1 -\log_2{A(r+\rho-1,\rho)}.$ Within a repair group, even if $t \le \rho-1$ coordinates are erased, they can be recovered by the rest of the coordinates. 
Therefore, for a repair group with $t< \rho-1$ coordinates erased, at least $\mu(r+\rho-1, \rho)-t$ redundant coordinates remain.

Which means, 
$$
H(Z_1^n \mid I) \le \avg_{\rm BEC}  \big(n - |I| -\frac{L_I}{r+\rho -1}\big), 
$$
where,  the subscript BEC denote that the average is with respect to the randomness in BEC, and
$$
L_I = \sum_{t=0}^{\mu(r+\rho-1, \rho)-1} J(t) (\mu(r+\rho-1, \rho) -t).
$$

Hence,
\begin{align*}
H(Z_1^n \mid I) & \le n - \avg_{\rm BEC} |I| - \frac{1}{r+\rho-1} \avg_{\rm BEC} L_I \\
& = n - np - \frac{1}{r+\rho -1} \avg_{\rm BEC} L_I \\
& =  n - np - \frac{1}{r+\rho -1}  \sum_{t=0}^{\mu(r+\rho-1, \rho)-1}   (\mu(r+\rho-1, \rho) -t) \avg_{\rm BEC} J(t).
\end{align*}

Let us now derive $\avg_{\rm BEC} J(t).$ Let $\chi_i$ be the indicator random variable that denotes that the repair group of $i$th coordinate has $t$ coordinates erased. 
We have 
$$
\Pr(\chi_i =1) = \binom{r+\rho-1}{t} p^t (1-p)^{r+\rho-1 -t} .
$$
Therefore, 
\begin{align*}
\avg_{\rm BEC} J(t)= n \binom{r+\rho-1}{t} p^t (1-p)^{r+\rho-1 -t} .
\end{align*}

Plugging this in, we have, 
$$
H(Z_1^n \mid I)  \le n - np - \frac{n}{r+\rho -1}  \sum_{t=0}^{\mu(r+\rho-1, \rho)-1}   (\mu(r+\rho-1, \rho) -t)  \binom{r+\rho-1}{t} p^t (1-p)^{r+\rho-1 -t} .
$$

 Therefore,
$$
P(\cC) \ge 1- \frac{1- p -\frac{(1-p)^{r}}{r+\rho-1} \sum_{t=0}^{\mu(r+\rho-1, \rho)-1}  (\mu(r+\rho-1, \rho) -t)  \binom{r+\rho-1}{t} p^t (1-p)^{\rho-1 -t} }{R}- \frac{1}{nR}.
$$
To achieve vanishing probability of error, one must have,
$$
R \le 1- p -\frac{(1-p)^{r}}{r+\rho-1} \sum_{t=0}^{\mu(r+\rho-1, \rho)-1}   (\mu(r+\rho-1, \rho) -t)  \binom{r+\rho-1}{t} p^t (1-p)^{\rho-1 -t} .
$$

\end{proof}

\paragraph{Example.} Consider $\rho =3$ in Theorem~\ref{thm:main2}. From the sphere packing bound (or the Hamming bound), 
$$
A(r+2, 3) \le \frac{2^{r+2}}{r+3},
$$
which is achieved by the Hamming codes, when they exist with the parameters.
Therefore, $\mu(r+2, 3) \ge \log_2(r+3)$. Suppose, $r=5$ (since a Hamming code exists with length $7$ and distance $3$).
Then,
\begin{equation}
\capa_{\rm BEC}(p, r=5, \rho=3) \le 1-p - \frac{(1-p)^5}{7}(3+8p+10p^2).
\end{equation}

\subsection{Achievability: Proof of Theorem~\ref{thm:main3}}
We show our achievability result by devising  a random code. The key idea is to use  codewords of a small local code as the repair groups, by considering them as symbols of some larger alphabet code. The formal proof is below.
\begin{proof}[Proof of Theorem~\ref{thm:main3}]
We will show this  by constructing a code. Let $\Delta \equiv r+\rho-1$. Partition the set of $n$ coordinates into $\frac{n}{\Delta}$ groups of size $\Delta$ each. We assume that $\Delta$ divides $n$. However, this assumption is not necessary, as we can neglect the last $< \Delta$ remainder coordinates, with only $<\frac{\Delta}{n} \to 0$ reduction in rate.  Now, consider the $\Delta$ bits of a group as a super-symbol. Consider the input-output channel induced by these super-symbols instead of the BEC. We find the capacity of this super-channel, and then normalize by $\Delta$.

To construct a code with $(\rho,r)$ locality we first choose a fixed  code $\cA$ of length $\Delta$ and distance $\rho$. Next we construct a random code $\cC$ of length $n$. A codeword $c = (c_1 | c_2 | \ldots |c_{\frac{n}{\Delta}})$ of $\cC$ is formed by concatenating $\frac{n}{\Delta}$ randomly and uniformly chosen codewords of $\cA$ side-by-side. 
This code can be thought of as a concatenated code~\cite{forney1966concatenated}, with $\cA$ as the inner code and a random code of length $\frac{n}{\Delta}$ and alphabet size $|\cA|$ as the outer code.
We can think of this random code being used over a discrete memoryless channel over the larger alphabet of super-symbols. For decoding, we employ a joint-typicality decoder that considers the each block of $\Delta$ bits as a super-symbol over an alphabet of size $2^\Delta$.
It is known that the rate of a random code,
such that the probability of error goes to zero is given by $I(X_{1}^{\Delta}; Y_{1}^{\Delta})$ (sufficient condition), where $X_{1}^{\Delta}$ is a randomly and uniformly chosen codeword of $\cA$ and $Y_{1}^{\Delta}$ is the output of a BEC with flip probability $p$ when the input to the BEC is  $X_{1}^{\Delta}$. However this rate of information is achieved by $\Delta$ uses of the binary-input channel. Therefore, the rate of the concatenated code $\cC$ that results in vanishing probability of error is: 
$$
\frac{1}{\Delta} I(X_{1}^{\Delta}; Y_{1}^{\Delta}),
$$
with $X_{1}^{\Delta},Y_{1}^{\Delta}$ defined as above.

Now we have,
\begin{align*}
\frac{1}{\Delta} I(X_{1}^{\Delta}; Y_{1}^{\Delta}) &= \frac{1}{\Delta} (H(Y_{1}^{\Delta}) - (\Delta) h(p))  \\
& =  \frac{1}{\Delta} H(Y_{1}^{\Delta}) - h(p).
\end{align*}
We can calculate $\Pr(Y_{1}^{\Delta}=y_{1}^{\Delta})$ when $\cA$ is a linear code. Suppose $E(u,v)$ be the number of subsets of $\{1, 2, \dots, r+\rho-1\}$ of size $u$, such that the generator matrix of $\cA$ restricted to only those subsets have rank $v$. 

Let within a repair group, the set of erasures induced by the BEC is $S \subseteq \{1, 2, \dots, \Delta\}$. Let $\cA_{\bar{S}}\subseteq \{0,1\}^{\Delta - |S|}$ denote the code $\cA$ restricted to only the coordinates of $\bar{S} \equiv \{1, 2, \dots, \Delta\} \setminus S$. Let $k_{\bar{S}}$ denote the dimension of  $\cA_{\bar{S}}$ or the rank of the generator matrix of $\cA$ restricted to $ \{1, 2, \dots, \Delta\} \setminus S$.

Let $y_{1}^{\Delta} \in \{0,1, ?\}^\Delta$ be a binary vector with coordinates $S \subseteq \{1, 2, \dots, \Delta\}$ erased (denoted by $?$).
\begin{align*}
\Pr(Y_{1}^{\Delta}=y_{1}^{\Delta}) =\begin{cases} p^{|S|}(1-p)^{\Delta - |S|} \frac{1}{|\cA_{\bar{S}}|}  & y_{\bar{S}} \in \cA_{\bar{S}}\\
0 & \text{ otherwise }
\end{cases}
\end{align*}

Therefore, 
\begin{align*}
H(Y_{1}^{\Delta}) & = - \sum_{S \subseteq \{1, \dots, \Delta\}} 2^{k_{\bar{S}}} p^{|S|}(1-p)^{\Delta - |S|} \frac{1}{2^{k_{\bar{S}}}} \log \Big( p^{|S|}(1-p)^{\Delta - |S|} \frac{1}{2^{k_{\bar{S}}}} \Big) \\
& = - \sum_{u, v} E(u,v)  p^{\Delta -u}(1-p)^{u} \log \Big( p^{\Delta -u}(1-p)^{u} \frac{1}{2^{v}}\Big) \\
& = - \sum_{u, v} E(u,v)  p^{\Delta -u}(1-p)^{u}   \Big( (\Delta-u) \log p +u \log (1-p) -v \Big) \\
& = - \sum_u \binom{\Delta}{u} p^{\Delta -u}(1-p)^{u}  \Big( (\Delta-u) \log p +u \log (1-p) \Big) \\
& \quad + \sum_{u, v} vE(u,v)  p^{\Delta -u}(1-p)^{u}  \\
& = - \Delta \log p + (1-p)\Delta  \log \frac{p}{1-p} + \sum_{u, v} vE(u,v)  p^{\Delta -u}(1-p)^{u}.
\end{align*}

Therefore, 
\begin{align*}
\frac{1}{\Delta} I(X_{1}^{\Delta}; Y_{1}^{\Delta}) &= h(p) + \frac{1}{\Delta}  \sum_{u, v} vE(u,v)  p^{\Delta -u}(1-p)^{u} -h(p) \\
&=  \frac{1}{\Delta}  \sum_{u, v} vE(u,v)  p^{\Delta -u}(1-p)^{u} \\
&= \frac{p^\Delta}{\Delta}\left. \frac{\partial}{\partial y} U\Big(\frac{1-p}{p},y\Big) \right|_{y=1} \\
&= \frac{p^{r+\rho-1}}{r+\rho-1}\left. \frac{\partial}{\partial y} U\Big(\frac{1-p}{p},y\Big) \right|_{y=1},
\end{align*}
where $U(x,y)$ is the rank polynomial of the code $\cA$.

\end{proof}
\paragraph{Example.}
While we have calculated the rank polynomial of the single parity-check code earlier, for more general codes it is difficult. However, for optimal distance $3$ codes, i.e., Hamming codes, it is possible to derive. Even that is quite cumbersome exercise, however here we outline the method.

The rank polynomials of a code and its dual code are related by a MacWilliams-type identity, see~\cite{barg2002some}. If $\cC\in \{0,1\}^n$ and $\cC^\perp\in \{0,1\}^n$ are dual codes of each other with dimensions $k$ and $n-k$, and  rank polynomials $U(x,y)$ and $U^\perp(x,y)$, respectively, then
$$
U^\perp(x,y) = x^n y^{n-k} U\Big(\frac{1}{xy},y\Big).
$$
The dual of Hamming code is Simplex code, for which the number of subsets of columns of generator matrix with given rank has been calculated in \cite{barg1997matroid} (Barg attributed the result to Laksov, 1965~\cite{laksov1965linear}). Plugging them in, and using the identity above, one obtains the rank polynomial for Hamming codes. Then using Theorem~\ref{thm:main3}, one can bound the generalized LRC capacity of BEC from below. We refrain from reproducing the long expressions here.

Instead it turns out that one can indeed find an expression to bound $\capa_{\rm BEC}(p, r, \rho)$ by using the concept of more capable channel, which is more amenable to analysis. We will see this in the subsequent sections. Furthermore, it turns out that the  methods of this section extends to other channels. This is what we attempt in the immediate next section.


\section{LRC Capacity of the Binary Symmetric Channel}\label{sec:bsc}
Recall that, for a binary symmetric channel with error probability $p$, the Shannon capacity is $1-h(p)$. Recall also that when we are constrained to use a locally recoverable code with locality parameters $r$ and $\rho$ as the code, the capacity is $\capa_{\rm BSC}(p,r,\rho)$. For clarity, we first present the usual case of $\rho=2$, followed by the results for general $\rho$, though the first is just a corollary of the later.
\begin{theorem}\label{thm:bsc}
The capacity of LRC with locality  $r$ over BSC($p$) follows:
$$
1 - h(p) - \frac1{r+1}\Big(1-h\Big(\frac{1-(1-2p)^{r+1}}{2}\Big)\Big) \le \capa_{\rm BSC}(p,r,2) \le 1- h(p) -\frac{(1-h(p))^{r+1}}{r+1}.
$$
\end{theorem}

The converse result for the general $\rho$ is given below. The upper bound for Theorem~\ref{thm:bsc} follows as a corollary.
\begin{theorem}\label{thm:bsc2}
Recall that the least possible redundancy of a code of length $n$ and distance $d$ is $\mu(n,d)$. It follows that,
$$
\capa_{\rm BSC}(p, r,\rho) \le  1- h(p) -\frac{(1-h(p))^{r}}{r+\rho-1} \sum_{t=0}^{\mu(r+\rho-1, \rho)-1}   (\mu(r+\rho-1, \rho)-t)  \binom{r+\rho-1}{t} h(p)^t (1-h(p))^{\rho-1 -t}. 
$$
\end{theorem}
Since $\mu(r+1, 2) =1$, substituting above we obtain the upper bound of Theorem~\ref{thm:bsc}. The proof of Theorem \ref{thm:bsc2} follows from the more general results about binary-input symmetric discrete memoryless channels. We postpone the proof till next section.

\paragraph{Example.} Since from the sphere-packing bound $\mu(r+2, 3) \ge \log_2(r+3)$, we find,
\begin{equation}\label{eq:ham}
\capa_{\rm BSC}(p, 5,3) \le 1-h(p) - \frac{(1-h(p))^5}{7}(3+8h(p)+10h(p)^2).
\end{equation}
 
Now, we provide the general achievability result.
\begin{theorem}\label{thm:bsc3}
Suppose there exists a linear code $\cA$ of length $\Delta \equiv r+\rho-1$ and minimum distance $\rho$. Let $A^{(i)}(x,y)$ be the $i$th coset weight enumerator polynomial of $\cA$, $i=0, 1,2, \ldots, \frac{2^{\Delta}}{|\cA|}-1$. Then 
$$
\capa_{\rm BSC}(p,r, \rho) \ge \frac{\log|\cA| + H\Big( \{A^{(i)}(1-p,p)\}_{i=0}^{\frac{2^{\Delta}}{|\cA|}-1} \Big)}{\Delta} - h(p).
$$
where, $H(\{p_i\}_{i=1}^\ell) \equiv - \sum_{i=1}^\ell p_i\log_2 p_i$ is the entropy function.
\end{theorem}

We will prove this result next.

\subsection{Achievability: Proof of Theorem~\ref{thm:bsc3}}
We first prove the achievability part of Theorem~\ref{thm:bsc} which will explain the intuition better. Restating the claim below:
\begin{proposition}
There exists a family of $(n, Rn, r)$ LRC codes with rate 
$$
R \ge 1 - h(p) - \frac1{r+1}\Big(1-h\Big(\frac{1-(1-2p)^{r+1}}{2}\Big)\Big),
$$
that when used over a BSC($p$) results in a probability of error that goes to $0$ with $n$.
\end{proposition}
\begin{proof}
We will show the above by constructing a code. Again, partition the set of $n$ coordinates into $\frac{n}{r+1}$ groups of size $r+1$ each. As earlier, we can simply neglect the remainder coordinates if $r+1 \nmid n$. Now, consider the $r+1$ bits of a group as a super-symbol. Consider the input-output channel induced by these super-symbols instead of the BSC. We find the capacity of this channel.

Let us choose the codewords in the following way. Within each group $r$ symbols are uniformly and independently (Bernoulli($1/2$)) chosen. The last symbol of each group is the modulo-2 sum of the other $r$ symbols. The rate of this code such that the probability of error being vanishing is given by
$$
\frac{1}{r+1} I(X_{1}^{r+1}; Y_{1}^{r+1}),
$$
where $X_{1}^{r+1}, Y_{1}^{r+1}$ represents the $r+1$-bit input and output. Note that we arrive at this rate by considering the group of $r+1$ bits as a supersymbol from an alphabet of size $2^r$, and using a joint-typicality decoder. Now we have,
\begin{align*}
\frac{1}{r+1} I(X_{1}^{r+1}; Y_{1}^{r+1}) &= \frac{1}{r+1} (H(Y_{1}^{r+1}) - (r+1) h(p)) \\
& =  \frac{1}{r+1} H(Y_{1}^{r+1}) - h(p).
\end{align*}
We can now calculate $\Pr(Y_{1}^{r+1}=y_{1}^{r+1}).$

\begin{align*}
\Pr(Y_{1}^{r+1}&=y_{1}^{r+1}) =  \sum_{x_{1}^{r+1}}\Pr(Y_{1}^{r+1}=y_{1}^{r+1} |X_{1}^{r+1}=x_{1}^{r+1} ) \Pr(X_{1}^{r+1}=x_{1}^{r+1})\\
& = \frac{1}{2^r} \sum_{x_{1}^{r+1}: \wt(x_{1}^{r+1}) \text{ is even }}\Pr(Y_{1}^{r+1}=y_{1}^{r+1} |X_{1}^{r+1}=x_{1}^{r+1} )\\
& = \frac{1}{2^r} \sum_{x_{1}^{r+1}: \wt(x_{1}^{r+1}) \text{ is even }} p^{\dist(x_{1}^{r+1}, y_{1}^{r+1})} (1-p)^{r+1 - \dist(x_{1}^{r+1}, y_{1}^{r+1})}\\
& = \begin{cases} 
\frac{1}{2^r}\sum_{w \text{ even }} \binom{r+1}{w} p^w (1-p)^{r+1-w}, \quad \text{ when } \wt(y_1^{r+1}) \text{ even }\\
\frac{1}{2^r}\sum_{w \text{ odd }} \binom{r+1}{w} p^w (1-p)^{r+1-w}, \quad \text{ when } \wt(y_1^{r+1}) \text{ odd }
\end{cases}\\
&= \begin{cases} 
\frac{1}{2^{r+1}}(1+(1-2p)^{r+1}), \quad \text{ when } \wt(y_1^{r+1}) \text{ even }\\
\frac{1}{2^{r+1}}(1-(1-2p)^{r+1}), \quad \text{ when } \wt(y_1^{r+1}) \text{ odd }
\end{cases}
\end{align*}
Therefore, 
\begin{align*}
H(Y_1^{r+1}) &= -\frac{2^r}{2^{r+1}}  (1+(1-2p)^{r+1}) \log \frac{(1+(1-2p)^{r+1})}{2^{r+1}}\\
& \qquad - \frac{2^r}{2^{r+1}}  (1-(1-2p)^{r+1}) \log \frac{(1-(1-2p)^{r+1})}{2^{r+1}}.
\end{align*}
After some simplifications, we have
$$
\frac{1}{r+1} I(X_{1}^{r+1}; Y_{1}^{r+1}) = 1-h(p) - \frac1{r+1}\Big(1-h\Big(\frac{1-(1-2p)^{r+1}}{2}\Big)\Big).
$$
\end{proof}


To extend the achievability result for general LRCs with $\rho >2$, we need to  ensure that the codewords restricted to each repair group  form a code with minimum distance $\rho$. Therefore it makes sense to choose random codewords of a code of distance $\rho$ as disjoint repair blocks to form the overall LRC. For this we need to figure out $H(Y_{1}^{r+\rho-1})$ where $Y_{1}^{r+\rho-1}$ is the output of a BEC where the input $X_{1}^{r+\rho-1}$ is a randomly chosen codeword of a fixed code $\cA$ of distance $\rho$. 
If $\cA$ is a linear code and the channel is BSC, then the entropy of the output of the channel can be computed if we know the {\em coset weight distribution} of the code.

\begin{proof}[Proof of Theorem~\ref{thm:bsc3}]
To construct a code with $(\rho,r)$ locality we first choose a fixed linear code $\cA$ of length $\Delta \equiv r+\rho-1$ and distance $\rho$. Next we construct a random code $\cC$ of length $n$. A codeword $c = (c_1 | c_2 | \ldots |c_{\frac{n}{\Delta}})$ of $\cC$ is formed by concatenating $\frac{n}{\Delta}$ randomly and uniformly chosen codewords of $\cA$ side-by-side. As earlier, we can assume that $n$ is divisible by $\Delta$ without loss of any generality. 
Again, if we use a joint-typicality decoding then the achievable rate of transmission is given by,
$$
\frac{1}{\Delta} I(X_{1}^{\Delta}; Y_{1}^{\Delta}),
$$
where $X_{1}^{\Delta}$ is a randomly and uniformly chosen codeword of $\cA$ and $Y_{1}^{\Delta}$ is the output of a BSC with flip probability $p$ when the input to the BSC is  $X_{1}^{\Delta}$.
Now we have,
\begin{align*}
\frac{1}{\Delta} I(X_{1}^{\Delta}; Y_{1}^{\Delta}) &= \frac{1}{\Delta} (H(Y_{1}^{\Delta}) - (\Delta) h(p)) \\
& =  \frac{1}{\Delta} H(Y_{1}^{\Delta}) - h(p).
\end{align*}
We can calculate $\Pr(Y_{1}^{\Delta}=y_{1}^{\Delta})$ when $\cA$ is a linear code.

\begin{align*}
\Pr(Y_{1}^{\Delta}=y_{1}^{\Delta}) &=  \sum_{x_{1}^{\Delta}}\Pr(Y_{1}^{\Delta}=y_{1}^{\Delta} |X_{1}^{\Delta}=x_{1}^{\Delta} ) \Pr(X_{1}^{\Delta}=x_{1}^{\Delta})\\
&= \frac{1}{|\cA|} \sum_{x_{1}^{\Delta} \in \cA}\Pr(Y_{1}^{\Delta}=y_{1}^{\Delta} |X_{1}^{\Delta}=x_{1}^{\Delta} )\\
& = \frac{1}{|\cA|} \sum_{x_{1}^{\Delta} \in \cA} p^{\dist(x_{1}^{\Delta},y_{1}^{\Delta})} (1-p)^{\Delta - \dist(x_{1}^{\Delta},y_{1}^{\Delta})}\\
& = \frac{1}{|\cA|} \sum_{w=0}^{\Delta} A^{(i)}_w p^w(1-p)^{\Delta-w},
\end{align*}
if $y_{1}^{\Delta}$ belongs to the $i$th coset of the code, where $A^{(i)}_w$ is the number of vectors of Hamming weight $w$ in the $i$th coset of the code $\cA$, $i=0,1,  2, \ldots, \frac{2^{\Delta}}{|\cA|}-1$. Recall the {\em coset weight enumerator} of the code is given by:
$$
A^{(i)}(x,y) = \sum_{w=0}^{\Delta} A^{(i)}_w  x^{\Delta-w} y^w.
$$
Then,
\begin{align*}
H(Y_{1}^{\Delta})& = - \sum_{i=0}^{\frac{2^{\Delta}}{|\cA|}-1} |\cA| \cdot \frac{1}{|\cA|} A^{(i)}(1-p,p) \log \frac{A^{(i)}(1-p,p)}{|\cA|}\\
& = - \sum_{i=0}^{\frac{2^{\Delta}}{|\cA|}-1}A^{(i)}(1-p,p) \log A^{(i)}(1-p,p) + \log |\cA| \sum_{i=0}^{\frac{2^{\Delta}}{|\cA|}-1}A^{(i)}(1-p,p) .
\end{align*}
Now,
\begin{align*}
\sum_{i=0}^{\frac{2^{\Delta}}{|\cA|}-1}A^{(i)}(1-p,p) &= \sum_{i=0}^{\frac{2^{\Delta}}{|\cA|}-1} \sum_{w=0}^{\Delta} A^{(i)}_w p^w(1-p)^{\Delta-w}\\
& = \sum_{w=0}^{\Delta} \Big(\sum_{i=0}^{\frac{2^{\Delta}}{|\cA|}-1} A^{(i)}_w \Big) p^w(1-p)^{\Delta-w}\\
& = \sum_{w=0}^{\Delta} \binom{\Delta}{w} p^w(1-p)^{\Delta-w} =1.
\end{align*}
Therefore, 
\begin{align*}
H(Y_{1}^{\Delta}) = \log|\cA| + H\Big( \{A^{(i)}(1-p,p)\}_{i=0}^{\frac{2^{\Delta}}{|\cA|}-1} \Big),
\end{align*}
where $H(\{p_i\}_{i=1}^\ell) \equiv - \sum_{i=1}^\ell p_i\log_2 p_i$.
Overall,
$$
\frac{1}{\Delta} I(X_{1}^{\Delta}; Y_{1}^{\Delta}) = \frac{\log|\cA| + H\Big( \{A^{(i)}(1-p,p)\}_{i=0}^{\frac{2^{\Delta}}{|\cA|}-1} \Big)}{\Delta} - h(p).
$$
\end{proof}
\paragraph{Example: Hamming code as local codes.}
By taking the code $\cA$ to be the Hamming code of length $\Delta=r+2$, we can therefore have the following result for $\rho =3$, as the coset-weight distribution of Hamming code is known~\cite{MS1977}:
\begin{align*}
\capa_{\rm BSC}(p,r, 3) &\ge  1-h(p) - \frac{1}{r+3}(1-(1-2p)^{\frac{r+3}{2}}) \log(1-(1-2p)^{\frac{r+3}{2}}) \\
&- \frac{1+(r+2)(1-2p)^{\frac{r+3}{2}}}{(r+2)(r+3)} \log(1+(r+2)(1-2p)^{\frac{r+3}{2}}).
\end{align*}
In particular, 
\begin{align}
\capa_{\rm BSC}(p,5, 3) &\ge 1-h(p) - \frac{1}{8}(1-(1-2p)^{4}) \log(1-(1-2p)^{4})\nonumber \\
&- \frac{1+7(1-2p)^{4}}{56} \log(1+7(1-2p)^{4}).
\end{align}
One can compare this with the converse bound of Eq.~\ref{eq:ham}. Note that the bounds coincide both when $p \to 0$ and $p \to \frac12$.

As will discussed in detail in the next section, this automatically gives a lower bound on $\capa_{\rm BEC}(p,r,3)$ since BEC is a more capable channel. 
\begin{align*}
\capa_{\rm BEC}(p,r,3) &\ge  1-p - \frac{1}{r+3}(1-(1-2h^{-1}(p))^{\frac{r+3}{2}}) \log(1-(1-2h^{-1}(p))^{\frac{r+3}{2}}) \\
&- \frac{1+(r+2)(1-2h^{-1}(p))^{\frac{r+2}{2}}}{(r+2)(r+3)} \log(1+(r+2)(1-2h^{-1}(p))^{\frac{r+3}{2}}).
\end{align*}
At $p=0$ this bound evaluates to $\capa_{\rm BEC}(p=0,r, 3) \ge 1- \frac{\log(r+3)}{r+2}$. Note that, from the upper bound we have,
$\capa_{\rm BEC}(p=0,\rho=3,r) \le 1- \frac{\log(r+3)}{r+2}$. Therefore the bounds are  tight  at $p=0$. Similar tightness can be observed as $p \to 1$.

\section{General binary input-symmetric channels}\label{sec:all}

The results for general binary input-symmetric channels follow from the converse and achievability results for BEC or BSC because in some sense these channels are the best and worst among the general cases respectively. In fact, the converse for BSC (Theorem~\ref{thm:bsc2}) also follows from this reasoning. To formalize this, we will use the notion of {\em more capable channel}. 
Since we have an impossibility (converse) result for BEC and an achievability result for BSC, using Prop.~\ref{prop:cap}, we can obtain the following result.
\begin{theorem}
For any binary-input symmetric discrete memoryless channel $W$,
$$
\capa(W) - \frac1{r+1}\Big(1-h\Big(\frac{1-(1-2h^{-1}(1-\capa(W)))^{r+1}}{2}\Big)\Big) \le \capa(W,r,2) \le \capa(W) -\frac{\capa(W)^{r+1}}{r+1}.
$$
\end{theorem}
\begin{proof}
For a channel $W$, suppose $\capa(W)=1-p$. Therefore, a BEC with erasure probability $p$ must be more capable than the channel $W$. 
There exists an LRC of rate $\capa(W,r,2)$ that achieves a vanishing probability of error over the channel $W$. Therefore, there exists an LRC of rate $\capa(W,r,2)$ that achieves a vanishing probability of error over the BEC of erasure probability $p$. This implies,
$$
\capa(W,r,2) \le 1-p -\frac{(1-p)^{r+1}}{r+1},
$$
which proves the upper bound.

On the other hand, suppose $\capa(W)=1-h(p').$ Therefore, a BSC with flip probability $p'$ must be less capable than the channel $W$. We know that there exists a code of rate 
$$
1-h(p') - \frac1{r+1}\Big(1-h\Big(\frac{1-(1-2p')^{r+1}}{2}\Big)\Big),
$$
that achieves a vanishing probability of error over the BSC with error probability $p'$. Therefore there must exist a code of same rate that achieves a vanishing probability of error over the channel $W$.
\end{proof}

Since we have upper and lower bounds for LRCs over BEC and BSC respectively for the general case of $\rho >2$, we can obtain bounds for general binary input discrete memoryless channels via similar argument. We refrain from writing those somewhat clumsy expressions here.


%
%

\section{Conclusion}
We have characterized the error-correcting capabilities of optimal locally recoverable codes when used in a setting of stochastic errors and erasures.
There are some compelling open problems left to study regarding capacity of LRCs. First of all, for a BSC, the gap to capacity is not exactly characterized for even $\rho=2$. We conjecture that the upper bound on the gap (see Table~\ref{tab:results}) is tight. 

It should be noted that LRCs have been generalized to facilitate multiple (disjoint) repair groups for each candidate~e.g.\cite{cai2019optimal,karingula2022lower,tamo2016bounds}. It will be of interest to see how the capacity scale with that requirement.

Finally, while we do not foresee an obstacle to extend the results for larger alphabets, it would be good to have them documented. 

\vspace{0.5in}

\noindent{\em Acknowledgement:} 
The author is grateful to Alexander Barg (for discussions on the rank polynomial), and Hamed Hassani and Chandra Nair (for discussions on the `more capable' channels).

\bibliographystyle{abbrv}
\bibliography{aryabib}
\end{document}